\documentclass[10pt,aps,pra,twocolumn,superscriptaddress,nofootinbib,longbibliography]{revtex4-2}

\usepackage[T1]{fontenc}
\usepackage[utf8]{inputenc}
\usepackage{amsmath,amssymb,amsthm,mathtools,bm}
\usepackage{hyperref}
\usepackage{graphicx}

\hypersetup{hidelinks,allbordercolors={0 0 0}}

\newtheorem{theorem}{Theorem}
\newtheorem{proposition}{Proposition}

\newcommand{\Tr}{\operatorname{Tr}}
\newcommand{\sgn}{\operatorname{sgn}}
\newcommand{\atanh}{\operatorname{atanh}}
\newcommand{\R}{\mathbb{R}}

\newcommand{\op}{\mathrm{op}}
\newcommand{\calK}{\mathcal{K}}

\begin{document}

\title{No Free Compression in Quantum Relaxations for Optimization}

\author{Stuart Hadfield}
\affiliation{USRA Research Institute for Advanced Computer Science, Moffett Field, California 94035, USA}
\email{shadfield@usra.edu}
\thanks{ORCID: \href{https://orcid.org/0000-0002-4607-3921}{0000-0002-4607-3921}}

\date{September 29, 2026}

\begin{abstract}
Qubit-efficient quantum relaxations compress classical decision variables into expectation values on substantially fewer qubits.  We ask what resource tradeoffs this compression entails for quantum optimization.  For the complete quadratic-Majorana encoding of $m=\Theta(n^2)$ binary variables on~$n$ qubits, we define the universal margin as the smallest correlator magnitude that can be guaranteed with prescribed signs for every target sign (bit) assignment.  We show that it is exactly $\Delta_{\rm Maj}(n)=\tan\!\left(\frac{\pi}{4n}\right)=\Theta(1/n)$, whereas uniformly random sign assignments retain $\Theta(1/\sqrt n)$ target-specific margins.  Arbitrary density operators and mixed fermionic Gaussian states generate the same quadratic-Majorana covariance body, so non-Gaussian state resources cannot enlarge this two-point relaxation.  More generally, standard quantum random access code bounds provide general information-theoretic baselines.  For any fixed family of $m$ designated binary observables on~$n$ qubits, the universal margin is at most $\sqrt{2\ln2\,n/m}$, and random access decoding from $N$ copies with constant success probability above $1/2$ requires $nN=\Omega(m)$.  For a fixed Pauli correlation encoding required to work uniformly over all targets, maintaining a fixed nonzero decoded magnitude under smooth sign decoding therefore requires a rescaling parameter that grows as the available margin shrinks.  Thus, while providing substantial qubit savings, compression can shift cost into restricted expectation value geometry, smaller expectation value magnitudes, or more demanding information recovery rather than eliminate it.
\end{abstract}

\maketitle

\emph{Introduction.} Quantum optimization spans exact, approximate, and heuristic approaches to numerous challenging application problems~\cite{Abbas2024QuantumOptimization}.  For a given algorithm, the problem formulation determines how an instance is represented and therefore how many qubits are required, while at the same time available quantum hardware imposes a strict independent limit on the number of usable qubits.  This mismatch is especially problematic for reaching utility-scale applications with conventional binary encodings used in approaches such as the quantum approximate optimization algorithm (QAOA)~\cite{Farhi2014QAOA,Hadfield2019QAOA} and quantum annealing~\cite{KadowakiNishimori1998QA,Farhi2001Adiabatic}, where each decision variable is commonly assigned to its own qubit, so that a problem with $m$ variables requires $m$ or more qubits.  Hybrid decomposition and iterative problem reduction methods attempt to mitigate this constraint by instead solving smaller subproblems with the quantum device~\cite{Ponce2025Decomposition,BradyHadfield2024IQA}.  Alternatively, recently proposed compressed encodings offer a complementary strategy by changing the representation itself so that classical decision variables are represented within a much smaller quantum register, significantly reducing the qubit requirement at fixed problem size.  Indeed, the prospect of tackling application-scale optimization problems well beyond the reach of direct qubit-per-variable encodings on foreseeable qubit-limited devices has motivated substantial interest in qubit-efficient encodings.  In particular, recent work has explored quantum relaxations based on expectation values~\cite{Tan2021QubitEfficient,Fuller2024QuantumRelaxations,Patti2022MBE,SundarDupont2026QubitEfficient,RaymondButsPistoia2025FewQubits}, though their advantages and tradeoffs are not yet well understood.  Likewise, standard discrete encodings already exhibit tradeoffs between qubit requirements and circuit complexity~\cite{Sawaya2023EncodingTradeoffs}, while expectation-value relaxations seek substantially stronger reductions in qubit register width.

Quantum random-access optimization (QRAO) realizes constant-factor compression using quantum random access codes (QRAC)~\cite{Fuller2024QuantumRelaxations}.  Standard $(2,1)$ and $(3,1)$ encodings place two or three bits in one qubit with recovery biases $1/\sqrt2$ and $1/\sqrt3$ (i.e., success probability $(1+b)/2$ for bias $b$)~\cite{Ambainis1999DenseCoding,Hayashi2006QRAC}.  Teramoto \emph{et al.} quantified an algorithmic compression-performance tradeoff for MaxCut, with proven approximation ratio bounds of $0.625$ and $0.555$ at compression factors of two and three, respectively~\cite{Teramoto2023QRAO}.  Whether small expectation value margins alone constrain approximation ratios remains open.  Our bounds instead concern representational geometry and information recovery.  Other work studies entanglement and state resources, recursion, constraints, hardware noise, and quantum-to-classical decoding~\cite{TeramotoRaymond2023Entanglement,Kondo2025RecursiveQRAO,SharmaRaymond2024Knapsack,Tamura2024NoiseQRAO,He2025NonvariationalQRAO,RaymondButsPistoia2025FewQubits,Suzuki2026DecoderConsistent}.  Decoder-consistent relaxations make readout part of the relaxation by defining a Hamiltonian whose expectation equals the expected decoded objective~\cite{Suzuki2026DecoderConsistent}, underscoring that compression and readout form one physical protocol.

Pauli correlation encodings (PCE) pursue more aggressive compression.  Throughout, $m$ is the number of logical binary variables and $n$ the number of qubits.  In the expectation value encodings considered here, each variable has one designated observable.  For a density operator $\rho$ and observable $O$, write $\langle O\rangle_\rho=\Tr(\rho O)$.  PCE assigns each $x_i\in\{\pm1\}$ to a Pauli-string observable $\Pi_i$ and decodes
\begin{equation}
    x_i=\sgn\langle \Pi_i\rangle_\rho,\qquad i=1,\ldots,m.
    \label{eq:pce}
\end{equation}
Fixed-weight $k$-body correlations can therefore encode $m=\Theta(n^k)$ variables on $n$ qubits, corresponding in principle to width $n=\Theta(m^{1/k})$~\cite{Sciorilli2025PCE}.  This parametrically stronger compression is attractive because it could bring much larger logical optimization problems within a fixed qubit budget, and it has motivated applications to the low-autocorrelation binary sequence (LABS) problem, the traveling-salesman problem, portfolio optimization, and unit commitment~\cite{Sciorilli2025LABS,doCarmo2026WarmPCE,Soloviev2026Portfolio,Nguyen2026UnitCommitment}.  Recent methodological work also shows that practical usefulness depends on more than width alone, including constraint handling and decoding, finite-shot effects and hardware noise, correlator resolution and binarization, and behavior in direct hardware comparisons with QRAO~\cite{Padin2026ConstrainedPCE,Friedhoff2026ProblemAwarePCE,Alonso2026Benchmark,Yoshioka2026LargeScalePCE,Sharma2026Hardware}.  Efficiently simulable PCE constructions based on free fermions and IQP circuits further show that large width savings and useful heuristic behavior do not by themselves imply quantum advantage~\cite{LizzioBosco2026SimulablePCE}.  The exact theorem below concerns the complete quadratic Majorana family generated by mutually anticommuting operators, not generic sparse or homogeneous fixed-weight Pauli families.

We ask whether a fixed compressed encoding can retain a margin bounded away from zero as the problem size grows, for every target bit string.  For a target sign $x\in\{\pm1\}$, the aligned margin is $x\langle O\rangle_\rho$, which is positive precisely when the expectation value has the target sign.  For the complete quadratic-Majorana family underlying the free-fermionic construction, let $\Delta_{\rm Maj}(n)$ be the largest $\delta$ such that every target sign pattern can be represented with every correlator having its target sign and magnitude at least $\delta$.  We prove
\begin{equation}
    \Delta_{\rm Maj}(n)=\tan\!\left(\frac{\pi}{4n}\right)
    \sim \frac{\pi}{4n}.
    \label{eq:headline-margin}
\end{equation}

Every target remains representable, and we show that mixed Gaussian states attain this universal guarantee.  For the complete quadratic-Majorana encoding, Eq.~\eqref{eq:headline-margin} shows that representing $\Theta(n^2)$ binary variables on $n$ qubits forces the worst-case margin down to $\Theta(1/n)$.  At fixed confidence, recovering one isolated worst-case correlator sign without further assumptions on the state then requires $\Theta(n^2)$ prepared copies.  Independent coordinatewise decoding must then somehow resolve or amplify that margin, while structured or global decoders must draw on additional information.

Our results compare one exact obstruction with two increasingly general information-theoretic baselines.  The complete quadratic-Majorana family has an exactly solvable worst-case margin below the generic scale.  For any fixed family of designated binary observables, QRAC theory bounds the universal margin, while Nayak's bound~\cite{Nayak1999QRAC} constrains the product $nN$ for arbitrary random-access encodings.  The two general bounds are standard information-theoretic limits recast as baselines for compressed optimization.  The exact Majorana result is parametrically stronger than these generic baselines.  Thus its geometry is model-specific, while the baselines quantify unavoidable costs for more general compressed encodings.  Related capacity limits concern full high-dimensional continuous distributions~\cite{Barthe2025ExpectationSamplers}.  Discrete optimization often requires only that the encoded signs can be recovered reliably.

Each target sign pattern defines a tournament, an orientation of the complete graph whose edge directions are set by the pairwise correlator signs.  The worst-case sign patterns (Theorem~\ref{thm:exact-majorana}) are exactly those equivalent, under switching (reversing every edge between a vertex subset and its complement) and relabeling, to the transitive tournament, the acyclic orientation in which every edge follows a total vertex ordering.  For $m=\Theta(n^2)$, this $\Theta(n^{-1})$ worst-case margin is asymptotically smaller than the general $O(n^{-1/2})$ margin bound, while typical sign patterns attain $\Theta(n^{-1/2})$.  Deng \emph{et al.} proved a sharp bound on the spectral norm of tournament skew-adjacency matrices~\cite{Deng2018Tournament}.  Using this extremum, we connect Majorana sign realizability to tournament theory and construct a matching rank-two certificate.  This yields the exact margin, its compression consequences, and the separation between worst-case and typical sign patterns.  Figure~\ref{fig:summary} summarizes this exact obstruction and the two general baselines.  Proofs and auxiliary results are provided in the Supplemental Material.

\begin{figure*}[t]
\centering
\includegraphics[width=0.94\textwidth]{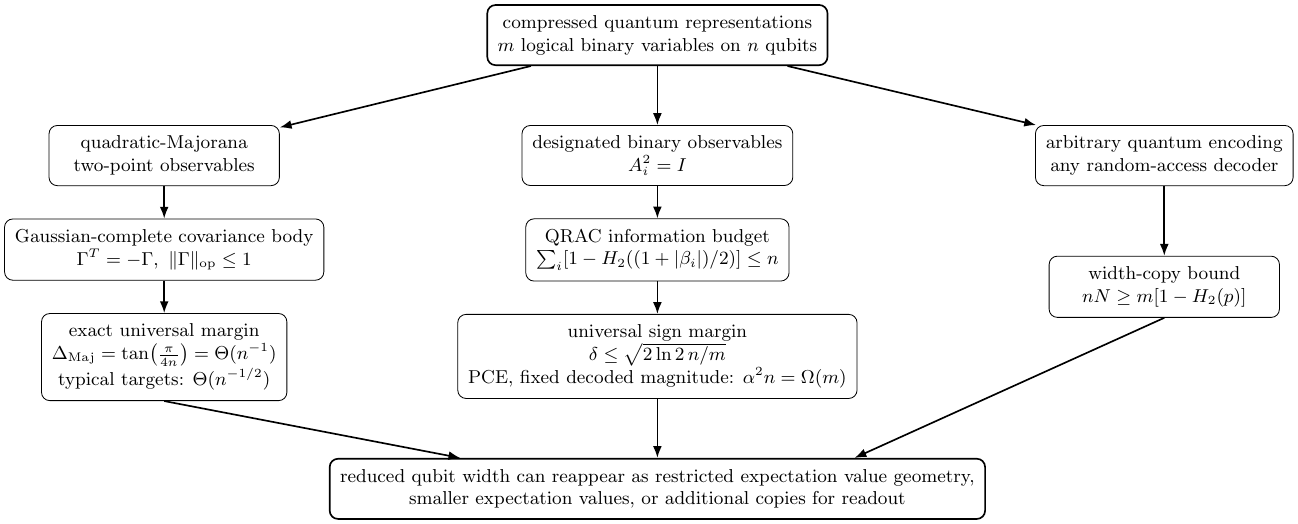}
\caption{Three levels of compression limits.  \emph{Majorana-specific.} The exact result gives a $\Theta(n^{-1})$ universal worst-case margin despite $\Theta(n^{-1/2})$ typical behavior, together with a Gaussian-complete set of attainable two-point correlations (the covariance body).  At fixed confidence, our worst-case margin result sets a requirement of $\Theta(n^2)$ state copies for independent recovery of a single expectation value sign, given only a lower bound on its expectation-value magnitude.  \emph{Designated observables.} Standard QRAC information theory bounds the universal margin for fixed binary observables.  \emph{Arbitrary random-access codes.} Nayak's bound converts width reduction into a copy cost for reliable recovery.}
\label{fig:summary}
\end{figure*}

\emph{Gaussian completeness.}
Let $\gamma_1,\ldots,\gamma_{2n}$ be Majorana operators satisfying $\{\gamma_a,\gamma_b\}=2\delta_{ab}I$, where $I$ is the identity operator.  For a density operator $\rho$, define the real antisymmetric covariance matrix
\begin{equation}
    \Gamma_{ab}(\rho)=i\Tr(\rho\gamma_a\gamma_b)\quad(a\ne b),
    \qquad \Gamma_{aa}=0.
    \label{eq:cov}
\end{equation}
The set of physical covariance matrices is the spectrahedron (a set defined by an affine linear matrix inequality)~\cite{RamanaGoldman1995SDP,Bravyi2005Lagrangian,Wan2023MatchgateShadows}
\begin{align}
    \calK_n
    &=\{\Gamma\in\R^{2n\times2n}:\Gamma^T=-\Gamma,\ \|\Gamma\|_{\op}\le1\}\notag\\
    &=\{\Gamma:\Gamma^T=-\Gamma,\ I+i\Gamma\succeq0\}.
    \label{eq:Kn}
\end{align}
Consider any finite family of $r$ Hermitian quadratic Majorana observables
\begin{equation}
    B_\ell=b_\ell I+\sum_{a<b}c_{\ell,ab}\,i\gamma_a\gamma_b,
    \qquad \ell=1,\ldots,r,
    \label{eq:quadratic-features}
\end{equation}
with real $b_\ell,c_{\ell,ab}$.  The $2n$ Majorana generators furnish an irreducible $n$-qubit representation of the complex Clifford algebra, unique up to unitary equivalence but with nonunique Pauli realizations~\cite{BravyiKitaev2002,Bravyi2005Lagrangian}.  The Jordan-Wigner transform maps $i\gamma_a\gamma_b$ to Pauli strings, while Bravyi-Kitaev gives a different realization with different Pauli-weight and locality properties~\cite{Seeley2012BravyiKitaev}.  Fermion-to-qubit mappings more generally trade Pauli weight against locality~\cite{Jiang2020Ternary}.  These representation-level costs are distinct from the expectation value geometry studied here.  The covariance body and margins depend only on Majorana anticommutation, while circuit and measurement costs can depend on the mapping.  The free-fermionic PCE construction uses precisely these quadratic covariance observables~\cite{LizzioBosco2026SimulablePCE}.  Ref.~\cite{ZhaoRubin2024FermionicEmbeddings} likewise represents orthogonal matrices by quadratic Majorana expectations, but uses them for a two-body fermionic relaxation and marginal-based rounding rather than to study universal margins for binary sign patterns.  Fermionic Gaussian (free-fermion) states are specified by their covariance matrices, with higher moments fixed by Wick's theorem~\cite{Bravyi2005Lagrangian}.

Quadratic Majorana observables are fermionic two-point correlators, and a mixed Gaussian state can reproduce every physical covariance matrix.

\begin{proposition}[Gaussian completeness]
\label{thm:gaussian}
For every density operator $\rho$ there exists a possibly mixed fermionic Gaussian state $\sigma$ satisfying $\Tr(\rho B_\ell)=\Tr(\sigma B_\ell)$ for all $\ell$.  Equivalently, arbitrary and Gaussian states generate the same affine image of $\calK_n$ under any quadratic observable map.
\end{proposition}

This result is standard covariance geometry~\cite{Bravyi2005Lagrangian,DiVincenzoTerhal2005FLO,Wan2023MatchgateShadows}.  If the objective and decoder depend only on the expectation values of the observables in Eq.~\eqref{eq:quadratic-features}, interactions may generate entanglement and non-Gaussianity without creating a new quadratic expectation value vector.  Consequently, objectives linear in these values can be optimized over $\calK_n$ by semidefinite programming.  The equality for attainable two-point expectation values also holds for nonlinear functions of those values, but does not make their optimization a semidefinite program.  Restricted preparation families may not reach all of $\calK_n$, and higher-order observables can carry non-Gaussian information not fixed by the covariance matrix~\cite{Bravyi2019ManyBody,HastingsODonnell2022}.

\emph{Exact margin of the complete Majorana relaxation.}
Take all
\begin{equation}
    A_{ab}=i\gamma_a\gamma_b,\qquad 1\le a<b\le2n,
\end{equation}
so $m=\binom{2n}{2}=n(2n-1)$.  For a target sign pattern $x=(x_{ab})$, define
\begin{equation}
    \delta_x=\max_\rho\min_{a<b}x_{ab}\Tr(\rho A_{ab}),
    \qquad
    \Delta_{\rm Maj}(n)=\min_x\delta_x,
    \label{eq:majorana-margin-def}
\end{equation}
where the maximization is over all density operators.

Since $\Delta_{\rm Maj}(n)>0$, every target sign pattern remains representable with a strictly positive margin.

\begin{theorem}[Exact Majorana margin]
\label{thm:exact-majorana}
For the complete quadratic-Majorana family,
\begin{equation}
    \Delta_{\rm Maj}(n)=\tan\!\left(\frac{\pi}{4n}\right).
    \label{eq:exact-majorana}
\end{equation}
The minimizing sign patterns are exactly those equivalent, under relabeling and switching, to the transitive tournament, the acyclic orientation of a complete graph in which every edge follows a total vertex ordering.  Hence $\Delta_{\rm Maj}(n)=\Theta(n^{-1})=\Theta(m^{-1/2})$.
\end{theorem}

The proof combines a covariance construction guaranteeing this margin for every target with a matching upper bound for a particular transitive target.  Set $d=2n$, the number of Majorana operators.  A target $x$ defines a tournament skew-adjacency matrix $T_x$ by $(T_x)_{ab}=x_{ab}$ for $a<b$ and $(T_x)_{ba}=-x_{ab}$.  The normalized tournament matrix
\begin{equation}
    \Gamma_x=\frac{T_x}{\|T_x\|_{\op}}
    \label{eq:cov-ray}
\end{equation}
is a valid covariance matrix and gives $\delta_x\ge1/\|T_x\|_{\op}$.  For a general target this construction need not be optimal.  Determining $\delta_x$ exactly for targets not equivalent to the transitive tournament under switching and relabeling remains open.  Deng \emph{et al.} proved
\begin{equation}
    \|T_x\|_{\op}\le \cot\!\left(\frac{\pi}{2d}\right),
    \label{eq:tournament-bound}
\end{equation}
with equality only for the transitive tournament up to switching and relabeling~\cite{Deng2018Tournament}.  For the transitive target, let $\theta=\pi/(2d)$ and choose positive normalized weights
\begin{equation}
    w_{ab}=\frac{2\tan\theta}{d}\sin\!\left(\frac{\pi(b-a)}{d}\right),
    \qquad \sum_{a<b}w_{ab}=1.
    \label{eq:certificate-weights}
\end{equation}
The corresponding skew coefficient matrix has rank two, with both nonzero singular values equal to $\tan\theta$.  The associated Hermitian quadratic Majorana Hamiltonian $H_w=\sum_{a<b}w_{ab}A_{ab}$ therefore has largest eigenvalue $\lambda_{\max}(H_w)=\tan\theta$.  Any state representing the transitive target with margin at least $\delta$, meaning $\Tr(\rho A_{ab})\ge\delta$ for every $a<b$, obeys $\delta\le\Tr(\rho H_w)\le\tan\theta$.  For this target, Eq.~\eqref{eq:cov-ray} and equality in Eq.~\eqref{eq:tournament-bound} give $\delta_x\ge1/\|T_{\rm tr}\|_{\op}=\tan\theta$, so the upper and lower bounds match and prove Eq.~\eqref{eq:exact-majorana}.  The optimal covariance $\Gamma_*=\tan\theta\,T_{\rm tr}$ has skew-block magnitudes $\tan\theta\,\cot\!\bigl((2j-1)\theta\bigr)$.  Thus one mode saturates the covariance bound, while the remaining fixed-index modes approach a $1/(2j-1)$ tail.  For this explicit optimum, $\sum_{a<b}\langle A_{ab}\rangle^2=m\tan^2\theta\to\pi^2/8$, even though the Frobenius norm envelope below allows a budget of $n$.  The stronger worst-case loss is therefore spectral, not implied by the global norm budget alone.

The worst-case behavior is not typical.  For every state,
\begin{equation}
    \sum_{a<b}\langle A_{ab}\rangle^2
    =\frac12\|\Gamma\|_F^2\le n,
    \label{eq:frobenius-envelope}
\end{equation}
so $\delta_x\le\sqrt{n/m}=1/\sqrt{2n-1}$ for every target.  For a uniformly random target, the tournament matrix has $\|T_x\|_{\op}=\Theta(\sqrt n)$ with high probability~\cite{Vershynin2018HDP}.  Equation~\eqref{eq:cov-ray} then gives the matching $\Omega(n^{-1/2})$ lower bound.  Hence $\delta_x=\Theta(n^{-1/2})$ typically, while the transitive switching class forces the universal $\Theta(n^{-1})$ bottleneck.  The theorem therefore identifies an exceptional but unavoidable worst-case family, while typical targets need not exhibit the stronger decay.  By comparison, the original PCE construction gives a general sufficient guarantee that every bit string is representable with correlator magnitudes $\Theta(1/m)$~\cite{Sciorilli2025PCE}.  That is a sufficiency construction, not an optimal-margin theorem.  The complete quadratic-Majorana family has exact worst-case margin $\Theta(m^{-1/2})$.  Whether other fixed $m=\Theta(n^2)$ families can approach the information-theoretic $O(n^{-1/2})=O(m^{-1/4})$ universal-margin envelope remains open.

Tournament matrices also enter the joint-measurement theory of the same Majorana family~\cite{McNulty2025FermionicJoint}.  There the relevant extremal quantity is the \emph{sum} of tournament singular values, whereas the state-space construction above is governed by the \emph{largest} singular value.  The transitive tournament is therefore margin-worst but minimizes skew energy, rather than maximizing the joint-measurement functional.  The Supplemental Material makes this distinction explicit.

\emph{General information limits beyond Majoranas.}
The exact result above exploits Majorana operator geometry.  We now remove that structure entirely.  Let $\mathbf A=(A_1,\ldots,A_m)$ be Hermitian binary observables on $n$ qubits, $A_i^2=I$, and let $x\mapsto\rho_x$ be any encoding of $x\in\{\pm1\}^m$.  Applying the standard entropic QRAC argument coordinate by coordinate~\cite{Nayak1999QRAC,Ambainis1999DenseCoding}, define for uniformly random $X=(X_1,\ldots,X_m)$ the average aligned bias
\begin{equation}
    \beta_i=2^{-m}\sum_x x_i\Tr(\rho_xA_i),
    \label{eq:beta}
\end{equation}
and let $H_2(q)=-q\log_2 q-(1-q)\log_2(1-q)$ denote the binary entropy.

\begin{proposition}[Coordinate information budget]
\label{thm:bias-budget}
Every such encoding satisfies
\begin{equation}
    \sum_{i=1}^{m}\left[1-H_2\!\left(\frac{1+|\beta_i|}{2}\right)\right]\le n,
    \label{eq:bias-budget}
\end{equation}
and therefore $\sum_i\beta_i^2\le2\ln2\,n$.
\end{proposition}

Proposition~\ref{thm:bias-budget} is the standard entropic QRAC bound specialized to designated binary observables.  Farkas \emph{et al.} give complementary dimension-based bounds on general QRAC average success~\cite{Farkas2025QRACBounds}.  If a universal margin $\delta$ exists, then $\beta_i\ge\delta$ for every $i$, so the geometric consequence below does not require these measurements actually to be used for decoding.  General multi-copy decoders enter only in Eq.~\eqref{eq:width-copies}.

For a fixed observable family $\mathbf A$, let
\begin{equation}
    \delta_x(\mathbf A)=\max_\rho\min_i x_i\Tr(\rho A_i),
    \qquad
    \Delta(\mathbf A)=\min_x\delta_x(\mathbf A),
    \label{eq:generic-margin}
\end{equation}
so $\delta_x$ is the margin for one target and $\Delta$ is the margin guaranteed uniformly over all targets.  With the probability simplex $\mathcal P_m=\{w\in\R_{\ge0}^m:\sum_i w_i=1\}$, Sion's minimax theorem gives~\cite{Sion1958Minimax}
\begin{equation}
    \delta_x(\mathbf A)=\min_{w\in\mathcal P_m}\lambda_{\max}\!\left(\sum_i w_i x_iA_i\right),
    \label{eq:minimax-main}
\end{equation}
an equivalent spectral characterization of the margin.  If $\Delta(\mathbf A)\ge\delta$ for $0\le\delta\le1$, Proposition~\ref{thm:bias-budget} implies
\begin{equation}
    n\ge m\left[1-H_2\!\left(\frac{1+\delta}{2}\right)\right],
    \qquad
    \delta\le\sqrt{\frac{2\ln2\,n}{m}}.
    \label{eq:generic-bound}
\end{equation}
Hence no fixed $n=o(m)$ binary-observable family can realize all sign patterns with constant margin.  This is a compression constraint: when $m\le n$, commuting Pauli $Z$ observables and product states give universal margin one.  Equation~\eqref{eq:generic-bound} is entropic and need not be sharp.  Recent work gives exact $(r,r-1)$ constructions~\cite{Suzuki2026QRAC} and proves conjectured average-success bounds for $(r,r-1)$ and $(r,r-2)$ QRACs~\cite{TanJafar2026QRAC}, while classical-code embeddings disprove a stronger square-root conjecture for unrestricted worst-case QRACs~\cite{Liu2026QRAC}.

Universal signs also arise naturally in optimization.  For an Ising objective $C$ with unique optimum $x^\star$, the spin-reversal gauges $C_s(x)=C(s\odot x)$ have optima $s\odot x^\star$ spanning the hypercube while preserving the interaction graph, absolute couplings, and spectrum.  Any \emph{fixed} observable family robust across this gauge family therefore obeys Eq.~\eqref{eq:generic-bound}.  This provides a concrete optimization reason for the universal-sign premise.  Gauge-adaptive or other problem-dependent mappings lie outside the statement.  The worst-case guarantee is thus operational for fixed hardware or compiler mappings, oblivious instance streams, and other claimed uniform guarantees.  Instance-dependent remapping can avoid a bad alignment but cannot change $\Delta(\mathbf A)$ once the observable family is fixed.

\emph{Where compression cost reappears in quantum optimization.}
An effective decoder ultimately needs stable classical decisions, not merely nonzero expectation values.  PCE therefore employs nonlinear maps of expectation values~\cite{Sciorilli2025PCE,Patti2022MBE} to turn small correlators into near-binary outputs, such as
\begin{equation}
    z_i=\tanh(\alpha\langle A_i\rangle),\qquad \alpha>0.
    \label{eq:decoder-tanh}
\end{equation}
Here $\alpha$ is a gain parameter.  Increasing it makes the transition around the sign threshold sharper but does not enlarge the underlying raw expectation value.
Related few-qubit work has also combined $\tanh$ activation with classical-shadow decoding~\cite{RaymondButsPistoia2025FewQubits}.  If $x_i z_i\ge c$ is required for every coordinate and target, with fixed $0<c<1$, the raw margin must be at least $\atanh(c)/\alpha$.  Equation~\eqref{eq:generic-bound} gives
\begin{equation}
    \alpha\ge \atanh(c)\sqrt{\frac{m}{2\ln2\,n}},
    \qquad \alpha^2n=\Omega(m).
    \label{eq:alpha-general}
\end{equation}
For $m=\Theta(n^k)$ this gives $\alpha=\Omega(n^{(k-1)/2})$.  Cubic compression already forces linear growth.  The exact Majorana theorem strengthens the quadratic case to
\begin{equation}
    \alpha\ge \atanh(c)\cot\!\left(\frac{\pi}{4n}\right)
    \sim \frac{4\atanh(c)}{\pi}\,n.
    \label{eq:alpha-majorana}
\end{equation}
The shrinking margin also sets the required measurement resolution, with the Majorana target in the worst case having an expectation value of order $1/n$.  Recovering this single expectation-value sign, given only a lower bound on its magnitude, then requires $\Omega(n^2)$ copies in the worst case, at fixed success probability even with collective measurements, while direct measurement achieves the same scaling (see Supplemental Material).  Writing $N$ for the number of prepared copies (the number of measurement shots for direct measurement), the worst-case Majorana coordinate therefore has a minimally sufficient rescaling $\alpha_{\min}=\Theta(n)$ and an isolated sign-recovery cost $N=\Theta(n^2)=\Theta(\alpha_{\min}^2)$.  Estimating all $m=\Theta(n^2)$ correlators separately at this precision and fixed error per coordinate gives a naive direct-measurement budget $\Theta(mN)=\Theta(m^2)$.  A guarantee on the probability that the entire recovered string is correct requires separate confidence accounting.  With destructive readout, this shot count also multiplies the state-preparation circuit cost.  Classical shadows, joint measurements, and structured decoders may instead amortize readout across many observables or infer bits from additional structure~\cite{Huang2020ClassicalShadows,RaymondButsPistoia2025FewQubits,Wan2023MatchgateShadows,McNulty2025FermionicJoint}.  Both costs above are consequences of the shrinking margin, consistent with PCE studies identifying correlator resolution and binarization as practical bottlenecks~\cite{Sciorilli2025PCE,Padin2026ConstrainedPCE,Yoshioka2026LargeScalePCE,Sharma2026Hardware,LizzioBosco2026SimulablePCE}.  A direct qubit-per-variable representation with $n=\Theta(m)$ can keep a constant signal per coordinate, so isolated direct readout requires only $O(1)$ copies at fixed error.

Noise couples directly to this accounting.  Under an effective attenuation model $\langle O\rangle\mapsto\eta\langle O\rangle$, the margin becomes $\Theta(\eta/n)$ and direct sign recovery scales as $\Theta(n^2/\eta^2)$.  Hence an additive calibration error or mitigation bias of order $1/n$ can instead flip a worst-case sign.  These are model-conditional consequences, not a universal hardware-noise lower bound~\cite{Tamura2024NoiseQRAO,Sharma2026Hardware}.

At the broadest level, if an arbitrary measurement on $N$ copies of $\rho_x$ returns any requested bit $x_i$ with worst-case success at least $p>1/2$, then $x\mapsto\rho_x^{\otimes N}$ is an $m$-bit QRAC on $nN$ qubits.  Nayak's bound~\cite{Nayak1999QRAC} gives
\begin{equation}
    nN\ge m[1-H_2(p)].
    \label{eq:width-copies}
\end{equation}
This applies to any random-access encoding and permits collective measurements across all copies.  Alternatively, QRAO magic-state rounding accesses the relaxed state differently from expectation value sign rounding and is not ruled out by a small coordinatewise margin.

\emph{Discussion.} The exact Majorana coding theorem and the two general information-theoretic bounds play different roles.  The $\Theta(n^{-1})$ worst-case margin is specific to quadratic Majorana geometry.  Gaussian completeness follows from standard covariance geometry and shows that non-Gaussian states do not enlarge the two-point relaxation.  The $O(\!\sqrt{n/m})$ observable bound and Nayak's width-copy bound are standard information-theoretic baselines.  Together, these results separate what is special to Majorana geometry from what follows from compression alone.  For the complete quadratic Majorana PCE family arising from the free-fermionic construction, the exact margin theorem has direct operational consequences: an $\alpha=\Omega(n)$ rescaling and an $\Omega(n^2)$ single-coordinate copy cost under only the expectation value magnitude promise, which provide sharp resource accounting benchmarks.  For other operator families, the minimax characterization and information-theoretic envelope provide diagnostics rather than a claimed scaling law.  Hence, when evaluating resource tradeoffs in quantum optimization applications, qubit compression should be assessed together with the observable family, the relevant margin scale, the decoder or rescaling rule, and the total measurement-shot budget.

Although the Hilbert space is exponentially large, the quadratic-Majorana relaxation accesses only a highly constrained set of two-point correlators.  Gaussian states exhaust the attainable covariance set, while the tournament correspondence identifies worst-case sign patterns with margin of order $1/n$.  Typical sign patterns retain the $n^{-1/2}$ scaling, so the theorem does not imply that typical compressed instances are particularly fragile.  Even this typical margin still shrinks with compression, however, so ``no free compression'' refers to resource tradeoff rather than failure of the encoding.

Nevertheless, useful schemes can fall outside these assumptions through algorithm modifications including problem-dependent observables, structured target sets, and global decoders that avoid the fixed observable family or random-access premises~\cite{Friedhoff2026ProblemAwarePCE,Suzuki2026DecoderConsistent,MaAngelakis2025InfoMinimal}.  The label--data encoding of Ref.~\cite{SundarDupont2026QubitEfficient} can retain conditional expectation values that do not shrink with problem size.  Only shots from the relevant label sector contribute to their estimation, so the sampling cost grows inversely with its probability.  Our exact Majorana theorem does not apply to this encoding.  Higher-order Majorana observables can probe non-Gaussian correlations invisible to two-point data~\cite{Bravyi2019ManyBody,HastingsODonnell2022,Tarabunga2026NonGaussianity}.  Continuous-domain PCE likewise lies outside the binary random-access premise~\cite{SolovievAdhikari2026Continuous}.  The exact theorem also does not directly cover sparse or homogeneous fixed-weight Pauli families, which lack the complete pairwise Majorana structure that yields the tournament correspondence.  Equation~\eqref{eq:minimax-main} gives a direct characterization of their worst-case margins.  It can also be applied to fixed families of quartic and higher-order Majorana observables, whose expectation values probe correlations beyond the covariance matrix.  Whether such practical sparse families follow the Majorana scaling or approach the generic information-theoretic envelope remains open.  No universal sparse-family scaling follows without specifying the observable family.

Qubit count alone therefore provides incomplete and potentially misleading resource accounting.  Register width reduction can carry hidden quantitative costs resulting from diminished margins, a restricted set of geometrically attainable expectation values, nonlinear rescaling, and additional state copies or measurement shots for readout.  The general bounds show that this tradeoff is not peculiar to Majoranas, while the exact result shows that operator geometry can magnify it.  Hence, the challenge lies not in compression itself, but in assessing resource tradeoffs holistically.  Complementary complexity barriers concern useful-state preparation and QRAO energy optimization~\cite{Hadfield2026StatePreparation,Hadfield2026QRAOComplexity}.  For quantum optimization, the resource advantages of compression must therefore be carefully assessed jointly in terms of qubit count, attainable margin, decoding rule, and measurement-shot budget, even before solution quality and overall algorithmic performance are considered.  Practical comparisons should also specify the target confidence, measurement grouping, and the noise model controlling the effective margin.

\paragraph*{Data availability.}
All results are analytical.  Numerical spot checks in the Supplemental Material are consistency checks reproducible from the stated minimax optimization.  No external data or separate research software are required.

\begin{acknowledgments}
The author is grateful for helpful discussions with Filip Maciejewski, Davide Venturelli, and Marco Sciorilli, and acknowledges support from AFRL Contract No. FA8750-25-C-B0040.  OpenAI ChatGPT was used under author-directed prompts for literature synthesis, cross-checks of derivations and claims, manuscript editing, and the drafting of \hyperref[fig:summary]{Fig.~\ref*{fig:summary}}.  The author assumes complete responsibility for all reported findings.
\end{acknowledgments}

\bibliography{references}

@article{Sciorilli2025PCE,
  author = {Sciorilli, Marco and Borges, Lucas and Patti, Taylor L. and Garc{\'i}a-Mart{\'i}n, Diego and Camilo, Giancarlo and Anandkumar, Anima and Aolita, Leandro},
  title = {Towards large-scale quantum optimization solvers with few qubits},
  journal = {Nature Communications},
  volume = {16},
  pages = {476},
  year = {2025},
  doi = {10.1038/s41467-024-55346-z}
}

@article{Fuller2024QuantumRelaxations,
  author = {Fuller, Bryce and Hadfield, Charles and Glick, Jennifer R. and Imamichi, Takashi and Itoko, Toshinari and Thompson, Richard J. and Jiao, Yang and Kagele, Marna M. and Blom-Schieber, Adriana W. and Raymond, Rudy and Mezzacapo, Antonio},
  title = {Approximate Solutions of Combinatorial Problems via Quantum Relaxations},
  journal = {IEEE Transactions on Quantum Engineering},
  volume = {5},
  pages = {1--15},
  year = {2024},
  doi = {10.1109/TQE.2024.3421294}
}

@article{LizzioBosco2026SimulablePCE,
  author = {Lizzio Bosco, Daniele and Matos, Gabriel and Liu, Chen-Yu and Rapp, Frederic and Finger, Fabian and Rinaldi, Enrico and Meichanetzidis, Konstantinos},
  title = {Efficiently Simulable {Pauli} Correlation Encoding},
  journal = {arXiv preprint arXiv:2607.20409},
  year = {2026},
  eprint = {2607.20409},
  archivePrefix = {arXiv},
  primaryClass = {quant-ph}
}

@article{Yoshioka2026LargeScalePCE,
  author = {Yoshioka, Takuya and Sasada, Keita and Usuki, Riku and Nakano, Yuichiro and Fujii, Keisuke},
  title = {Qubit-Efficient Variational Quantum Optimization via {Pauli} Correlation Encoding: Application to Large-Scale Power Demand Portfolio Optimization},
  journal = {Physical Review Applied},
  year = {2026},
  eprint = {2607.24722},
  archivePrefix = {arXiv},
  primaryClass = {quant-ph},
  doi = {10.1103/568y-4264}
}

@article{Barthe2025ExpectationSamplers,
  author = {Barthe, Alice and Grossi, Michele and Vallecorsa, Sofia and Tura, Jordi and Dunjko, Vedran},
  title = {Parameterized quantum circuits as universal generative models for continuous multivariate distributions},
  journal = {npj Quantum Information},
  volume = {11},
  pages = {121},
  year = {2025},
  doi = {10.1038/s41534-025-01064-3}
}

@article{Deng2018Tournament,
  author = {Deng, Bo and Li, Xueliang and Shader, Bryan and So, Wasin},
  title = {On the maximum skew spectral radius and minimum skew energy of tournaments},
  journal = {Linear and Multilinear Algebra},
  volume = {66},
  number = {7},
  pages = {1434--1441},
  year = {2018},
  doi = {10.1080/03081087.2017.1357676}
}

@inproceedings{Nayak1999QRAC,
  author = {Nayak, Ashwin},
  title = {Optimal lower bounds for quantum automata and random access codes},
  booktitle = {Proceedings of the 40th Annual Symposium on Foundations of Computer Science},
  pages = {369--376},
  year = {1999},
  doi = {10.1109/SFFCS.1999.814608}
}

@inproceedings{Ambainis1999DenseCoding,
  author = {Ambainis, Andris and Nayak, Ashwin and Ta-Shma, Amnon and Vazirani, Umesh},
  title = {Dense quantum coding and a lower bound for 1-way quantum automata},
  booktitle = {Proceedings of the 31st Annual ACM Symposium on Theory of Computing},
  pages = {376--383},
  year = {1999},
  doi = {10.1145/301250.301347}
}

@article{Farkas2025QRACBounds,
  author = {Farkas, M{\'a}t{\'e} and Miklin, Nikolai and Tavakoli, Armin},
  title = {Simple and general bounds on quantum random access codes},
  journal = {Quantum},
  volume = {9},
  pages = {1643},
  year = {2025},
  doi = {10.22331/q-2025-02-25-1643}
}

@article{Bravyi2005Lagrangian,
  author = {Bravyi, Sergey},
  title = {Lagrangian representation for fermionic linear optics},
  journal = {Quantum Information and Computation},
  volume = {5},
  number = {3},
  pages = {216--238},
  year = {2005},
  doi = {10.26421/QIC5.3-3},
  eprint = {quant-ph/0404180},
  archivePrefix = {arXiv}
}

@article{DiVincenzoTerhal2005FLO,
  author = {DiVincenzo, David P. and Terhal, Barbara M.},
  title = {Fermionic linear optics revisited},
  journal = {Foundations of Physics},
  volume = {35},
  number = {12},
  pages = {1967--1984},
  year = {2005},
  doi = {10.1007/s10701-005-8657-0}
}

@article{Huang2020ClassicalShadows,
  author = {Huang, Hsin-Yuan and Kueng, Richard and Preskill, John},
  title = {Predicting many properties of a quantum system from very few measurements},
  journal = {Nature Physics},
  volume = {16},
  pages = {1050--1057},
  year = {2020},
  doi = {10.1038/s41567-020-0932-7}
}

@article{Wan2023MatchgateShadows,
  author = {Wan, Kianna and Huggins, William J. and Lee, Joonho and Babbush, Ryan},
  title = {Matchgate shadows for fermionic quantum simulation},
  journal = {Communications in Mathematical Physics},
  volume = {404},
  pages = {629--700},
  year = {2023},
  doi = {10.1007/s00220-023-04844-0}
}

@article{McNulty2025FermionicJoint,
  author = {McNulty, Daniel and Calegari, Susane and Oszmaniec, Micha{\l}},
  title = {Optimal Fermionic Joint Measurements for Estimating Non-Commuting {Majorana} Observables},
  journal = {Quantum},
  volume = {9},
  pages = {1914},
  year = {2025},
  doi = {10.22331/q-2025-11-17-1914}
}

@article{Soloviev2026Portfolio,
  author = {Soloviev, Vicente P. and Krompiec, Michal},
  title = {Large-scale portfolio optimization using {Pauli} correlation encoding},
  journal = {Scientific Reports},
  volume = {16},
  pages = {25158},
  year = {2026},
  doi = {10.1038/s41598-026-54244-2}
}

@article{Alonso2026Benchmark,
  author = {Alonso, Fernando and Sampr{\'o}n, Colom{\'a}n and Veiga, Jacobo and Mussa Juane, Mariamo and G{\'o}mez, Andr{\'e}s},
  title = {Benchmark of {Pauli} Correlation Encoding for different optimisation problems},
  journal = {arXiv preprint arXiv:2606.18914},
  year = {2026},
  eprint = {2606.18914},
  archivePrefix = {arXiv},
  primaryClass = {quant-ph}
}

@article{Padin2026ConstrainedPCE,
  author = {Pad{\'i}n-Mart{\'i}nez, Jacobo and Soloviev, Vicente P. and Borrallo-Rentero, Alejandro and Rodr{\'i}guez-Otero, Ant{\'o}n and Alfonso-Rodr{\'i}guez, Raquel and Krompiec, Michal},
  title = {Progressive Binarization -- {Pauli} Correlation Encoding: a Continuation Method for Constrained Optimization},
  journal = {arXiv preprint arXiv:2602.17479},
  year = {2026},
  eprint = {2602.17479},
  archivePrefix = {arXiv},
  primaryClass = {quant-ph}
}

@article{Tarabunga2026NonGaussianity,
  author = {Tarabunga, Poetri Sonya and Jobst, Bernhard and Morral-Yepes, Ra{\'u}l and Langer, Marc and Kraus, Barbara and Pollmann, Frank and Lin, Sheng-Hsuan},
  title = {Computable fermionic non-{G}aussianity from the covariance matrix},
  journal = {arXiv preprint arXiv:2607.02242},
  year = {2026},
  eprint = {2607.02242},
  archivePrefix = {arXiv},
  primaryClass = {quant-ph}
}

@article{Friedhoff2026ProblemAwarePCE,
  author = {Friedhoff, Triet and Metkar, Mihir and Davis, Wade and Kumar, Vaibhaw and Galda, Alexey},
  title = {{Pauli} Correlation Encoding for {mRNA} Secondary Structure Prediction: Problem-Aware Decoding for Dense-Constraint {QUBO}s},
  journal = {arXiv preprint arXiv:2605.20163},
  year = {2026},
  eprint = {2605.20163},
  archivePrefix = {arXiv},
  primaryClass = {quant-ph}
}

@article{Patti2022MBE,
  author = {Patti, Taylor L. and Kossaifi, Jean and Anandkumar, Anima and Yelin, Susanne F.},
  title = {Variational quantum optimization with multibasis encodings},
  journal = {Physical Review Research},
  volume = {4},
  pages = {033142},
  year = {2022},
  doi = {10.1103/PhysRevResearch.4.033142}
}

@article{Tan2021QubitEfficient,
  author = {Tan, Benjamin and Lemonde, Marc-Antoine and Thanasilp, Supanut and Tangpanitanon, Jirawat and Angelakis, Dimitris G.},
  title = {Qubit-efficient encoding schemes for binary optimisation problems},
  journal = {Quantum},
  volume = {5},
  pages = {454},
  year = {2021},
  doi = {10.22331/q-2021-05-04-454}
}

@article{Hayashi2006QRAC,
  author = {Hayashi, Masahito and Iwama, Kazuo and Nishimura, Harumichi and Raymond, Rudy and Yamashita, Shigeru},
  title = {$(4,1)$-quantum random access coding does not exist---one qubit is not enough to recover one of four bits},
  journal = {New Journal of Physics},
  volume = {8},
  pages = {129},
  year = {2006},
  doi = {10.1088/1367-2630/8/8/129}
}

@article{Teramoto2023QRAO,
  author = {Teramoto, Kosei and Raymond, Rudy and Wakakuwa, Eyuri and Imai, Hiroshi},
  title = {Quantum-Relaxation Based Optimization Algorithms: Theoretical Extensions},
  journal = {arXiv preprint arXiv:2302.09481},
  year = {2023},
  eprint = {2302.09481},
  archivePrefix = {arXiv},
  primaryClass = {quant-ph},
}

@article{Kondo2025RecursiveQRAO,
  author = {Kondo, Ruho and Sato, Yuki and Raymond, Rudy and Yamamoto, Naoki},
  title = {Recursive Quantum Relaxation for Combinatorial Optimization Problems},
  journal = {Quantum},
  volume = {9},
  pages = {1594},
  year = {2025},
  doi = {10.22331/q-2025-01-15-1594}
}

@article{Tamura2024NoiseQRAO,
  author = {Tamura, Kentaro and Suzuki, Yohichi and Raymond, Rudy and Watanabe, Hiroshi C. and Sato, Yuki and Kondo, Ruho and Sugawara, Michihiko and Yamamoto, Naoki},
  title = {Noise Robustness of Quantum Relaxation for Combinatorial Optimization},
  journal = {IEEE Transactions on Quantum Engineering},
  volume = {5},
  pages = {3103009},
  year = {2024},
  doi = {10.1109/TQE.2024.3439135}
}

@article{He2025NonvariationalQRAO,
  author = {He, Zichang and Raymond, Rudy and Shaydulin, Ruslan and Pistoia, Marco},
  title = {Non-variational quantum random access optimization with alternating operator ansatz},
  journal = {Scientific Reports},
  volume = {15},
  pages = {29191},
  year = {2025},
  doi = {10.1038/s41598-025-13543-w}
}

@article{Suzuki2026QRAC,
  author = {Suzuki, Takayuki},
  title = {Analytical construction of $(n,n-1)$ quantum random access codes saturating the conjectured bound},
  journal = {Physical Review A},
  volume = {114},
  pages = {012441},
  year = {2026},
  doi = {10.1103/fcz5-2g3w}
}

@article{Liu2026QRAC,
  author = {Liu, Kangqiao},
  title = {Classical codes violate the conjectured square-root bound for quantum random access codes},
  journal = {arXiv preprint arXiv:2607.15617},
  year = {2026},
  eprint = {2607.15617},
  archivePrefix = {arXiv},
  primaryClass = {quant-ph}
}

@inproceedings{HastingsODonnell2022,
  author = {Hastings, Matthew B. and O'Donnell, Ryan},
  title = {Optimizing Strongly Interacting Fermionic {Hamiltonians}},
  booktitle = {Proceedings of the 54th Annual ACM SIGACT Symposium on Theory of Computing},
  pages = {776--789},
  year = {2022},
  doi = {10.1145/3519935.3519960}
}

@article{Bravyi2019ManyBody,
  author = {Bravyi, Sergey and Gosset, David and K{\"o}nig, Robert and Temme, Kristan},
  title = {Approximation algorithms for quantum many-body problems},
  journal = {Journal of Mathematical Physics},
  volume = {60},
  pages = {032203},
  year = {2019},
  doi = {10.1063/1.5085428}
}

@book{Vershynin2018HDP,
  author = {Vershynin, Roman},
  title = {High-Dimensional Probability: An Introduction with Applications in Data Science},
  publisher = {Cambridge University Press},
  year = {2018},
  doi = {10.1017/9781108231596}
}

@article{Sion1958Minimax,
  author = {Sion, Maurice},
  title = {On general minimax theorems},
  journal = {Pacific Journal of Mathematics},
  volume = {8},
  number = {1},
  pages = {171--176},
  year = {1958},
  doi = {10.2140/pjm.1958.8.171}
}

@article{RamanaGoldman1995SDP,
  author = {Ramana, Motakuri V. and Goldman, Alan J.},
  title = {Some geometric results in semidefinite programming},
  journal = {Journal of Global Optimization},
  volume = {7},
  pages = {33--50},
  year = {1995},
  doi = {10.1007/BF01100204}
}

@article{MaAngelakis2025InfoMinimal,
  author = {Ma, Gordon and Angelakis, Dimitris G.},
  title = {An Information-Minimal Geometry for Qubit-Efficient Optimization},
  journal = {arXiv preprint arXiv:2511.08362},
  year = {2025},
  eprint = {2511.08362},
  archivePrefix = {arXiv},
  primaryClass = {quant-ph}
}

@article{Sciorilli2025LABS,
  author = {Sciorilli, Marco and Camilo, Giancarlo and Maciel, Thiago O. and Canabarro, Askery and Borges, Lucas and Aolita, Leandro},
  title = {A competitive {NISQ} and qubit-efficient solver for the {LABS} problem},
  journal = {arXiv preprint arXiv:2506.17391},
  year = {2025},
  eprint = {2506.17391},
  archivePrefix = {arXiv},
  primaryClass = {quant-ph}
}

@article{Nguyen2026UnitCommitment,
  author = {Nguyen, Kien X. and Safro, Ilya and Liu, Xiaoyuan},
  title = {Scaling Quantum Optimization for Unit Commitment via {Pauli} Correlation Encoding},
  journal = {arXiv preprint arXiv:2605.17145},
  year = {2026},
  eprint = {2605.17145},
  archivePrefix = {arXiv},
  primaryClass = {quant-ph}
}

@article{Sharma2026Hardware,
  author = {Sharma, Monit and Lau, Hoong Chuin},
  title = {From Circuits to Hardware: Benchmarking Standard and Qubit-Efficient Quantum Optimization on Real Hardware},
  journal = {Quantum Science and Technology},
  year = {2026},
  doi = {10.1088/2058-9565/ae94a4},
  eprint = {2607.11637},
  archivePrefix = {arXiv},
  primaryClass = {quant-ph}
}

@article{Suzuki2026DecoderConsistent,
  author = {Suzuki, Takayuki},
  title = {Decoder-Consistent {Hamiltonians} for {POVM}-Based Quantum Relaxations},
  journal = {arXiv preprint arXiv:2606.05604},
  year = {2026},
  eprint = {2606.05604},
  archivePrefix = {arXiv},
  primaryClass = {quant-ph}
}

@article{SolovievAdhikari2026Continuous,
  author = {Soloviev, Vicente P. and Adhikari, Bibhas},
  title = {Quantum Learning of Classical Correlations with continuous-domain {Pauli} Correlation Encoding},
  journal = {arXiv preprint arXiv:2604.05637},
  year = {2026},
  eprint = {2604.05637},
  archivePrefix = {arXiv},
  primaryClass = {quant-ph}
}

@article{TanJafar2026QRAC,
  author = {Tan, Shuo and Jafar, Syed A.},
  title = {Optimal Average Success Probabilities of Binary $(n,n-1)$ and $(n,n-2)$ Quantum Random Access Codes via a Proof of the Corresponding Conjectured Bound},
  journal = {arXiv preprint arXiv:2607.10414},
  year = {2026},
  eprint = {2607.10414},
  archivePrefix = {arXiv},
  primaryClass = {quant-ph}
}

@article{doCarmo2026WarmPCE,
  author = {do Carmo, Rafael Sim{\~o}es and dos Reis, Renato Gomes and Silva, Samuel Fernando F. and Arruda, Luiz Gustavo E. and Fanchini, Felipe F.},
  title = {Warm-Starting {PCE} for Traveling Salesman Problem},
  journal = {Brazilian Journal of Physics},
  volume = {56},
  pages = {49},
  year = {2026},
  doi = {10.1007/s13538-025-01966-9},
  eprint = {2509.14414},
  archivePrefix = {arXiv},
  primaryClass = {quant-ph}
}

@article{ZhaoRubin2024FermionicEmbeddings,
  author = {Zhao, Andrew and Rubin, Nicholas C.},
  title = {Expanding the reach of quantum optimization with fermionic embeddings},
  journal = {Quantum},
  volume = {8},
  pages = {1451},
  year = {2024},
  doi = {10.22331/q-2024-08-28-1451},
  eprint = {2301.01778},
  archivePrefix = {arXiv},
  primaryClass = {quant-ph}
}

@article{Seeley2012BravyiKitaev,
  author = {Seeley, Jacob T. and Richard, Martin J. and Love, Peter J.},
  title = {The {Bravyi--Kitaev} Transformation for Quantum Computation of Electronic Structure},
  journal = {The Journal of Chemical Physics},
  volume = {137},
  pages = {224109},
  year = {2012},
  doi = {10.1063/1.4768229},
  eprint = {1208.5986},
  archivePrefix = {arXiv},
  primaryClass = {quant-ph}
}

@article{SundarDupont2026QubitEfficient,
  author = {Sundar, Bhuvanesh and Dupont, Maxime},
  title = {Qubit-efficient quantum combinatorial-optimization solver},
  journal = {Physical Review Applied},
  volume = {25},
  pages = {034071},
  year = {2026},
  doi = {10.1103/s5jv-jh24},
  eprint = {2407.15539},
  archivePrefix = {arXiv},
  primaryClass = {quant-ph}
}

@inproceedings{TeramotoRaymond2023Entanglement,
  author = {Teramoto, Kosei and Raymond, Rudy and Imai, Hiroshi},
  title = {The Role of Entanglement in Quantum-Relaxation Based Optimization Algorithms},
  booktitle = {2023 IEEE International Conference on Quantum Computing and Engineering (QCE)},
  volume = {1},
  pages = {543--553},
  year = {2023},
  doi = {10.1109/QCE57702.2023.00068},
  eprint = {2302.00429},
  archivePrefix = {arXiv},
  primaryClass = {quant-ph}
}

@inproceedings{SharmaRaymond2024Knapsack,
  author = {Sharma, Monit and Jin, Yan and Lau, Hoong Chuin and Raymond, Rudy},
  title = {Quantum Relaxation for Solving Multiple Knapsack Problems},
  booktitle = {2024 IEEE International Conference on Quantum Computing and Engineering (QCE)},
  volume = {1},
  pages = {692--698},
  year = {2024},
  doi = {10.1109/QCE60285.2024.00086},
  eprint = {2404.19474},
  archivePrefix = {arXiv},
  primaryClass = {quant-ph}
}

@inproceedings{RaymondButsPistoia2025FewQubits,
  author = {Raymond, Rudy and Buts, Alexander and Pistoia, Marco},
  title = {Combinatorial Optimization with Few Qubits},
  booktitle = {2025 IEEE International Conference on Quantum Computing and Engineering (QCE)},
  volume = {1},
  pages = {36--47},
  year = {2025},
  doi = {10.1109/QCE65121.2025.00015}
}

@article{Abbas2024QuantumOptimization,
  author = {Abbas, Amira and others},
  title = {Challenges and opportunities in quantum optimization},
  journal = {Nature Reviews Physics},
  volume = {6},
  pages = {718--735},
  year = {2024},
  doi = {10.1038/s42254-024-00770-9}
}

@article{Sawaya2023EncodingTradeoffs,
  author = {Sawaya, Nicolas P. D. and Schmitz, Albert T. and Hadfield, Stuart},
  title = {Encoding trade-offs and design toolkits in quantum algorithms for discrete optimization: coloring, routing, scheduling, and other problems},
  journal = {Quantum},
  volume = {7},
  pages = {1111},
  year = {2023},
  doi = {10.22331/q-2023-09-14-1111},
  eprint = {2203.14432},
  archivePrefix = {arXiv},
  primaryClass = {quant-ph}
}

@misc{Farhi2014QAOA,
  author = {Farhi, Edward and Goldstone, Jeffrey and Gutmann, Sam},
  title = {A Quantum Approximate Optimization Algorithm},
  year = {2014},
  eprint = {1411.4028},
  archivePrefix = {arXiv},
  primaryClass = {quant-ph}
}

@article{KadowakiNishimori1998QA,
  author = {Kadowaki, Tadashi and Nishimori, Hidetoshi},
  title = {Quantum Annealing in the Transverse {Ising} Model},
  journal = {Physical Review E},
  volume = {58},
  pages = {5355--5363},
  year = {1998},
  doi = {10.1103/PhysRevE.58.5355}
}

@article{Ponce2025Decomposition,
  author = {Ponce, Moises and Herrman, Rebekah and Lotshaw, Phillip C. and Powers, Sarah S. and Siopsis, George and Humble, Travis S. and Ostrowski, James},
  title = {Graph decomposition techniques for solving combinatorial optimization problems with variational quantum algorithms},
  journal = {Quantum Information Processing},
  volume = {24},
  pages = {60},
  year = {2025},
  doi = {10.1007/s11128-025-04675-z}
}

@article{BradyHadfield2024IQA,
  author = {Brady, Lucas T. and Hadfield, Stuart},
  title = {Iterative quantum algorithms for maximum independent set},
  journal = {Physical Review A},
  volume = {110},
  pages = {052435},
  year = {2024},
  doi = {10.1103/PhysRevA.110.052435}
}

@article{BravyiKitaev2002,
  author = {Bravyi, Sergey B. and Kitaev, Alexei Yu.},
  title = {Fermionic Quantum Computation},
  journal = {Annals of Physics},
  volume = {298},
  pages = {210--226},
  year = {2002},
  doi = {10.1006/aphy.2002.6254}
}

@article{Jiang2020Ternary,
  author = {Jiang, Zhang and Kalev, Amir and Mruczkiewicz, Wojciech and Neven, Hartmut},
  title = {Optimal fermion-to-qubit mapping via ternary trees with applications to reduced quantum states learning},
  journal = {Quantum},
  volume = {4},
  pages = {276},
  year = {2020},
  doi = {10.22331/q-2020-06-04-276},
  eprint = {1910.10746},
  archivePrefix = {arXiv},
  primaryClass = {quant-ph}
}

@article{Hadfield2019QAOA,
  author = {Hadfield, Stuart and Wang, Zhihui and O'Gorman, Bryan and Rieffel, Eleanor G. and Venturelli, Davide and Biswas, Rupak},
  title = {From the Quantum Approximate Optimization Algorithm to a Quantum Alternating Operator Ansatz},
  journal = {Algorithms},
  volume = {12},
  number = {2},
  pages = {34},
  year = {2019},
  doi = {10.3390/a12020034}
}

@article{Farhi2001Adiabatic,
  author = {Farhi, Edward and Goldstone, Jeffrey and Gutmann, Sam and Lapan, Joshua and Lundgren, Andrew and Preda, Daniel},
  title = {A Quantum Adiabatic Evolution Algorithm Applied to Random Instances of an {NP}-Complete Problem},
  journal = {Science},
  volume = {292},
  number = {5516},
  pages = {472--475},
  year = {2001},
  doi = {10.1126/science.1057726}
}

@article{Hadfield2026QRAOComplexity,
  author = {Hadfield, Stuart},
  title = {Complexity Amplification from Compression in Quantum Random Access Optimization},
  journal = {arXiv preprint arXiv:2609.07090},
  year = {2026},
  eprint = {2609.07090},
  archivePrefix = {arXiv},
  primaryClass = {quant-ph},
  doi = {10.48550/arXiv.2609.07090}
}

@article{Hadfield2026StatePreparation,
  author = {Hadfield, Stuart},
  title = {Complexity Barriers to State Preparation in Quantum Approximate Optimization},
  journal = {arXiv preprint arXiv:2609.31520},
  year = {2026},
  eprint = {2609.31520},
  archivePrefix = {arXiv},
  primaryClass = {quant-ph},
  doi = {10.48550/arXiv.2609.31520}
}

\clearpage
\onecolumngrid
\setcounter{equation}{0}
\setcounter{section}{0}
\setcounter{subsection}{0}
\setcounter{proposition}{0}
\renewcommand{\theequation}{S\arabic{equation}}
\renewcommand{\thesection}{S\Roman{section}}
\renewcommand{\thesubsection}{\thesection.\Alph{subsection}}
\renewcommand{\theproposition}{S\arabic{proposition}}
\renewcommand{\theHequation}{supp.\arabic{equation}}
\renewcommand{\theHsection}{supp.\Roman{section}}
\renewcommand{\theHsubsection}{supp.\Roman{section}.\Alph{subsection}}
\renewcommand{\theHproposition}{supp.\arabic{proposition}}
\begin{center}
{\large\bfseries SUPPLEMENTAL MATERIAL}
\end{center}
\label{sec:supplemental}

{
This Supplemental Material provides the technical details behind the Letter.  We use the same notation as the main text.  $n$ is the number of qubits and $m$ is the number of logical binary variables.  For the expectation value encodings considered here, $m$ is also the number of designated observables.  For the complete quadratic-Majorana family, $m=\binom{2n}{2}$.  Where useful, definitions are repeated so that individual arguments are self-contained.

\section{Fermionic covariance body and Gaussian completeness}
\label{supp:gaussian}
Let $\gamma_1,\ldots,\gamma_{2n}$ satisfy $\{\gamma_a,\gamma_b\}=2\delta_{ab}I$, where $I$ is the identity operator.  On $n$ qubits, these generators furnish an irreducible representation of the complex Clifford algebra $\mathrm{Cl}_{2n}(\mathbb C)$, unique up to unitary equivalence.  Changing the fermion-to-qubit mapping, for example from Jordan-Wigner to Bravyi-Kitaev, changes the Pauli-string realization and can change Pauli-weight and locality properties, but not the attainable covariance body or any margin defined from it~\cite{Bravyi2005Lagrangian,Seeley2012BravyiKitaev}.  For a density operator $\rho$, let
\begin{equation}
  \Gamma_{ab}=i\Tr(\rho\gamma_a\gamma_b)\quad(a\neq b),\qquad \Gamma_{aa}=0.
  \label{eq:s-cov}
\end{equation}
The matrix $\Gamma$ is real and antisymmetric.  Every real antisymmetric matrix can be brought by an orthogonal matrix $O$ to
\begin{equation}
    O\Gamma O^T=\bigoplus_{j=1}^{n}
    \begin{pmatrix}0&\lambda_j\\-\lambda_j&0\end{pmatrix}.
    \label{eq:s-canonical-Gamma}
\end{equation}
The signs of the $\lambda_j$ may be absorbed into the canonical blocks so that $O$ can be chosen in $SO(2n)$.  The transformed Majoranas $\tilde\gamma_a=\sum_bO_{ab}\gamma_b$ obey the same anticommutation relations, with $\lambda_j=\Tr\!\left(\rho\,i\tilde\gamma_{2j-1}\tilde\gamma_{2j}\right)$.  Because $i\tilde\gamma_{2j-1}\tilde\gamma_{2j}$ is Hermitian and squares to $I$, $|\lambda_j|\le1$.  Hence every physical covariance matrix belongs to
\begin{equation}
  \calK_n=\{\Gamma\in\R^{2n\times2n}:\Gamma^T=-\Gamma,\ \|\Gamma\|_{\op}\le1\}.
  \label{eq:s-Kn}
\end{equation}
Here $\|\cdot\|_{\op}$ denotes the operator norm.  Equivalently, $I+i\Gamma\succeq0$.

Conversely, let $\Gamma\in\calK_n$ and use the canonical form in Eq.~\eqref{eq:s-canonical-Gamma}.  For each canonical mode define
\begin{equation}
    \rho_j=\frac12\left(I_j+\lambda_j\,i\tilde\gamma_{2j-1}\tilde\gamma_{2j}\right),
    \label{eq:s-single-mode-gaussian}
\end{equation}
where $I_j$ is the identity on that mode.  Since $|\lambda_j|\le1$, $\rho_j$ is a valid one-mode thermal Gaussian state.  The product state $\bigotimes_j\rho_j$ has covariance $O\Gamma O^T$ in the transformed Majorana basis.  The $SO(2n)$ basis change is implemented by a fermionic Gaussian unitary.  Transforming back therefore yields a Gaussian state with covariance $\Gamma$~\cite{Bravyi2005Lagrangian,DiVincenzoTerhal2005FLO,Wan2023MatchgateShadows}.  Thus arbitrary density operators and fermionic Gaussian states generate exactly the same quadratic covariance body.

For any Hermitian quadratic observable with real coefficients,
\begin{equation}
  B_\ell=b_\ell I+\sum_{a<b}c_{\ell,ab}\,i\gamma_a\gamma_b,
  \qquad
  \langle B_\ell\rangle=b_\ell+\sum_{a<b}c_{\ell,ab}\Gamma_{ab}.
\end{equation}
Therefore the set of attainable vectors of observable expectation values $(\langle B_\ell\rangle)_\ell$ is the same for arbitrary density operators and for Gaussian states.  This proves Proposition~1 of the Letter.  The conclusion concerns representability.  Optimizing a nonlinear objective over the spectrahedron may still be hard.

\section{Exact universal margin for complete quadratic-Majorana observables}
\label{supp:exact-margin}
Set $d=2n$ and $\theta=\pi/(2d)=\pi/(4n)$, where $d$ is the number of Majorana operators.  For a target sign pattern $x=(x_{ab})_{a<b}$, define
\begin{equation}
    \delta_x=\max_\rho\min_{a<b}x_{ab}\Tr(\rho A_{ab}),\qquad A_{ab}=i\gamma_a\gamma_b.
    \label{eq:s-deltax}
\end{equation}
The maximum is over density operators on the $n$-qubit Hilbert space.
Associate to $x$ the real skew-symmetric tournament matrix $T_x$ defined by
\begin{equation}
    (T_x)_{ab}=x_{ab}\ (a<b),\qquad (T_x)_{ba}=-x_{ab},\qquad (T_x)_{aa}=0.
    \label{eq:s-Tx}
\end{equation}
Switching a vertex subset means reversing every tournament edge between that subset and its complement.  In matrix form it is $T_x\mapsto DT_xD$ with diagonal $D_{aa}\in\{\pm1\}$.  Relabeling is permutation similarity.  Because $T_x/\|T_x\|_{\op}\in\calK_n$, it is a valid covariance matrix and therefore gives
\begin{equation}
    \delta_x\ge \frac{1}{\|T_x\|_{\op}}.
    \label{eq:s-target-lower}
\end{equation}
This feasible ray gives a lower bound for every target, but it need not be optimal for arbitrary $x$.  Determining $\delta_x$ exactly for sign patterns not equivalent to the transitive tournament under switching and relabeling remains open.

Deng \emph{et al.} proved that the skew spectral radius of any tournament matrix of order $d$ satisfies~\cite{Deng2018Tournament}
\begin{equation}
    \|T_x\|_{\op}\le\cot\theta,
    \label{eq:s-tournament}
\end{equation}
with equality only, up to switching and relabeling, for the transitive tournament (equivalently, the acyclic tournament $x_{ab}=+1$ for $a<b$ in a suitable vertex ordering).  Hence every target obeys $\delta_x\ge\tan\theta$.

\subsection{Matching upper bound for the transitive target}
For the transitive target $x_{ab}=+1$ for all $a<b$, define positive weights
\begin{equation}
    w_{ab}=\frac{2\tan\theta}{d}\sin\!\left(\frac{\pi(b-a)}{d}\right),
    \qquad a<b.
    \label{eq:s-pweights}
\end{equation}
They are normalized because
\begin{equation}
\sum_{a<b}\sin\!\left(\frac{\pi(b-a)}{d}\right)
=\sum_{h=1}^{d-1}(d-h)\sin\!\left(\frac{\pi h}{d}\right) =\frac{d}{2}\cot\theta.
\label{eq:s-sin-sum}
\end{equation}
The last identity follows, for example, by differentiating a finite geometric series.

Let $W$ be the skew coefficient matrix
\begin{equation}
    W_{ab}=\frac{2\tan\theta}{d}\sin\!\left(\frac{\pi(b-a)}{d}\right)
    \qquad(a,b=1,\ldots,d),
    \label{eq:s-P}
\end{equation}
with $W_{aa}=0$ automatically.  Set $u_a=\cos(\pi a/d)$ and $v_a=\sin(\pi a/d)$.  Then
\begin{equation}
    W=\frac{2\tan\theta}{d}(uv^T-vu^T).
    \label{eq:s-ranktwo}
\end{equation}
The vectors $u$ and $v$ are orthogonal and satisfy $\|u\|^2=\|v\|^2=d/2$.  Therefore $W$ has rank two and its two nonzero singular values are both $\tan\theta$.

With $\gamma=(\gamma_1,\ldots,\gamma_d)^T$, define
\begin{equation}
    H_w=\sum_{a<b}w_{ab}\,i\gamma_a\gamma_b
    =\frac{i}{2}\gamma^TW\gamma.
    \label{eq:s-Hw}
\end{equation}
If $s_1,\ldots,s_n\ge0$ are the $n$ nonnegative skew-block magnitudes of a real skew matrix, the same orthogonal change of Majorana basis puts the corresponding quadratic Hamiltonian in the form
\begin{equation}
    H=\sum_{j=1}^{n}s_j\,i\tilde\gamma_{2j-1}\tilde\gamma_{2j}.
\end{equation}
The $n$ terms commute, and each $i\tilde\gamma_{2j-1}\tilde\gamma_{2j}$ has eigenvalues $\pm1$.  Let $\lambda_{\max}(H)$ denote the largest eigenvalue of $H$.  Hence
\begin{equation}
    \operatorname{spec}(H)=\left\{\sum_{j=1}^{n}\epsilon_j s_j:\epsilon_j\in\{\pm1\}\right\},
    \qquad \lambda_{\max}(H)=\sum_{j=1}^{n}s_j.
    \label{eq:s-quadratic-spectrum}
\end{equation}
For $H_w$ only one canonical value is nonzero, $s_1=\tan\theta$, so $\lambda_{\max}(H_w)=\tan\theta$.  If $\rho$ represents the transitive target with $\Tr(\rho A_{ab})\ge\delta$ for every $a<b$, positivity and normalization of the weights yield
\begin{equation}
    \delta\le\sum_{a<b}w_{ab}\Tr(\rho A_{ab})
    =\Tr(\rho H_w)\le\tan\theta.
\end{equation}
Together with Eq.~\eqref{eq:s-target-lower}, this proves $\delta_x=\tan\theta$ for the transitive target and hence Theorem~1 of the Letter.

\subsection{Uniqueness of the worst sign patterns}
If $T_x$ is not equivalent to the transitive tournament under switching and relabeling, the uniqueness part of Ref.~\cite{Deng2018Tournament} gives $\|T_x\|_{\op}<\cot\theta$.  Equation~\eqref{eq:s-target-lower} then gives $\delta_x>\tan\theta$, so such a target cannot attain the universal minimum.  Conversely, switching and relabeling act by signed-permutation similarity on $T_x$.  The covariance body $\calK_n$ is invariant under the same signed permutations, so $\delta_x$ is unchanged.  Every sign pattern obtained from the transitive tournament by switching and relabeling therefore attains $\tan\theta$.

\section{Frobenius bound and typical sign patterns}
\label{supp:typical}
For any physical covariance matrix $\Gamma\in\calK_n$, the $n$ parameters $\lambda_j$ in its canonical $2\times2$ block form lie in $[-1,1]$.  Here $\|\cdot\|_F$ denotes the Frobenius norm.  Therefore $\|\Gamma\|_F^2=2\sum_{j=1}^{n}\lambda_j^2\le2n$, and hence
\begin{equation}
    \sum_{a<b}\langle A_{ab}\rangle^2
    =\sum_{a<b}\Gamma_{ab}^2
    =\frac12\|\Gamma\|_F^2\le n.
    \label{eq:s-frob}
\end{equation}
If all $m=\binom{2n}{2}$ aligned coordinates have magnitude at least $\delta$, Eq.~\eqref{eq:s-frob} gives
\begin{equation}
    \delta\le\sqrt{\frac{n}{m}}=\frac{1}{\sqrt{2n-1}}.
    \label{eq:s-frob-margin}
\end{equation}
This purely geometric bound reproduces the general $n^{-1/2}$ scale, with a slightly sharper constant for this family.  The exact transitive-tournament result $\tan\!\left(\frac{\pi}{4n}\right)=\Theta(n^{-1})$ therefore requires information beyond the Frobenius constraint.

For a uniformly random target $x$, the entries $(T_x)_{ab}$ above the diagonal are independent Rademacher signs.  Write $T_x=U-U^T$, where $U$ is strictly upper triangular.  Its random entries are independent, mean-zero, sub-Gaussian variables (the remaining entries are deterministic zeros), so Theorem~4.4.5 of Ref.~\cite{Vershynin2018HDP} gives $\|U\|_{\op}=O(\sqrt d)$ with high probability.  Hence $\|T_x\|_{\op}\le2\|U\|_{\op}=O(\sqrt d)$.  The deterministic inequality $\|T_x\|_{\op}\ge\|T_x\|_F/\sqrt d=\sqrt{d-1}$ gives the matching $\Omega(\sqrt d)$ lower bound.  Thus $\|T_x\|_{\op}=\Theta(\sqrt n)$ with high probability.  Combining the ray lower bound in Eq.~\eqref{eq:s-target-lower} with Eq.~\eqref{eq:s-frob-margin} yields
\begin{equation}
    \delta_x=\Theta(n^{-1/2})
    \qquad\text{with high probability over a random sign pattern.}
    \label{eq:s-typical}
\end{equation}
This establishes the typical-versus-worst separation used in the Letter: typical sign patterns live at the information-theoretic scale, whereas the exceptional family obtained from the transitive tournament by switching and relabeling forces the universal $n^{-1}$ bottleneck.

\section{Coordinate information budget}
\label{supp:bias}
Let $A_1,\ldots,A_m$ be the designated Hermitian binary observables, with $A_i^2=I$.  Let $Z=(Z_1,\ldots,Z_m)$ be uniformly distributed on $\{\pm1\}^m$, and let $Q$ denote the $n$-qubit system prepared in state $\rho_Z$.  Define the average coordinate bias $\beta_i=\mathbb E_Z[Z_i\Tr(\rho_ZA_i)]=2^{-m}\sum_{z\in\{\pm1\}^m}z_i\Tr(\rho_zA_i)$.
For each $i$, measure $A_i$ and, if necessary, reverse the labels of its two outcomes so that the resulting classical variable $Y_i$ has nonnegative average correlation with $Z_i$.  Its average success probability is $p_i=\Pr[Y_i=Z_i]=(1+|\beta_i|)/2$.  Here $H_2(q)=-q\log_2 q-(1-q)\log_2(1-q)$ is the binary entropy, $I(U:V)$ denotes mutual information in bits, and $S(Q)$ is the von Neumann entropy in bits.  Data processing gives $I(Z_i:Q)\ge I(Z_i:Y_i)$.  Since $Z_i$ is unbiased and binary, Fano's inequality gives
\begin{equation}
    I(Z_i:Y_i)\ge1-H_2(p_i)
    =1-H_2\!\left(\frac{1+|\beta_i|}{2}\right).
    \label{eq:s-fano}
\end{equation}
Because the $Z_i$ are independent,
\begin{equation}
    I(Z:Q)
    =\sum_{i=1}^{m} I(Z_i:Q\mid Z_{<i}) =\sum_{i=1}^{m} I(Z_i:Q,Z_{<i}) \ge\sum_{i=1}^{m}I(Z_i:Q),
    \label{eq:s-superadditive-info}
\end{equation}
where $Z_{<i}=(Z_1,\ldots,Z_{i-1})$, the second equality uses $I(Z_i:Z_{<i})=0$, and the inequality is data processing.  The Holevo bound then gives $I(Z:Q)\le S(Q)\le n$.  Summing Eq.~\eqref{eq:s-fano} proves Proposition~2 of the Letter.

For $t\in[-1,1]$,
\begin{equation}
    1-H_2\!\left(\frac{1+t}{2}\right)
    =\frac{1}{\ln2}D_{\rm KL}\!\left(\mathrm{Ber}\!\left(\frac{1+t}{2}\right)\middle\|\mathrm{Ber}\!\left(\frac12\right)\right)
    \ge\frac{t^2}{2\ln2},
    \label{eq:s-pinsker}
\end{equation}
where $D_{\rm KL}(P\|Q)$ is relative entropy with natural logarithms and $\mathrm{Ber}(r)$ is the Bernoulli distribution of mean $r$.  The final step is Pinsker's inequality.  Applying Eq.~\eqref{eq:s-pinsker} coordinatewise gives $\sum_i\beta_i^2\le2\ln2\,n$.

If an encoding has pointwise designated-observable margin $x_i\Tr(\rho_xA_i)\ge\delta$ for every $x,i$, with $0\le\delta\le1$, then $\beta_i\ge\delta$ and
\begin{equation}
    n\ge m\left[1-H_2\!\left(\frac{1+\delta}{2}\right)\right],
    \qquad
    \delta\le\sqrt{\frac{2\ln2\,n}{m}}.
    \label{eq:s-generic-bound}
\end{equation}
This is Nayak's QRAC lower bound in the present sign-margin language~\cite{Nayak1999QRAC}.

\section{Spin-reversal gauge family}
\label{supp:gauge}
Let $C(x)=\sum_{i<j}J_{ij}x_ix_j+\sum_i h_i x_i$, with $J_{ij},h_i\in\R$, be an Ising objective with unique optimum $x^\star$.  For $s\in\{\pm1\}^m$, define $C_s(x)=C(s\odot x)$, where $\odot$ denotes componentwise multiplication.  If $x_s^\star$ optimizes $C_s$, then $s\odot x_s^\star=x^\star$, hence $x_s^\star=s\odot x^\star$.  The map $s\mapsto s\odot x^\star$ is a bijection of the hypercube.  Hence a fixed observable family that has pointwise margin at least $\delta$ at the optimum of every gauge provides a universal robust encoding and obeys Eq.~\eqref{eq:s-generic-bound}.

The fixed-observable-family assumption is essential.  A solver allowed to gauge-transform its observable assignment along with the instance uses a problem-dependent encoding and lies outside the statement.  The gauge construction therefore shows why the universal-margin premise is relevant to optimization.  A fixed encoding can already be forced to cover every optimum sign pattern by an isospectral, graph-preserving family of instances.

\section{Nonlinear decoders and rescaling}
\label{supp:rescaling}
The hyperbolic tangent in Eq.~\eqref{eq:decoder-tanh} of the Letter is not special.  For $\alpha>0$, let $f_\alpha:\R\to[-1,1]$ be odd and nondecreasing, and suppose a decoder requires $x_i f_\alpha(\Tr(\rho_xA_i))\ge c$ for all $x,i$, where $0<c<1$.  Define the finite positive threshold $r_\alpha(c)=\inf\{r>0:f_\alpha(r)\ge c\}$.  Universal decoding requires raw margin at least $r_\alpha(c)$.  Therefore $r_\alpha(c)\le\sqrt{2\ln2\,n/m}$ for arbitrary designated binary observables, and $r_\alpha(c)\le\tan(\pi/4n)$ for the complete quadratic-Majorana family.  For $f_\alpha(t)=\tanh(\alpha t)$, $r_\alpha(c)=\atanh(c)/\alpha$, giving the nonlinear-rescaling bounds in the Letter.

The original homogeneous PCE construction proves universal sign representability with correlators of magnitude $\Theta(1/m)$ and empirically uses a growing $\alpha$ to move the $\tanh$ response out of its nearly linear regime when correlators are small~\cite{Sciorilli2025PCE}.  Related few-qubit QRAC work also studies $\tanh$ activation and classical-shadow decoding~\cite{RaymondButsPistoia2025FewQubits}.  The information bound does not assert that a particular empirical schedule is optimal for that operator family.  It shows instead that the rescaling parameter must grow under uniform fixed-magnitude decoding, and that cubic compression already enforces linear growth in $n$ at the information-theoretic level.

\section{Width-copy tradeoff and single-coordinate sign recovery}
\label{supp:readout}
We separate a decoder-independent random-access statement from the cost of resolving the sign of one binary-observable expectation from copies of the state.

\begin{proposition}[Width-copy bound]
Let $x\mapsto\rho_x$ encode $m$ bits into $n$ qubits.  Suppose that, for each requested coordinate $i$, there exists a POVM on $N$ copies $\rho_x^{\otimes N}$ that returns $x_i$ with worst-case success probability at least $p>1/2$.  Then
\begin{equation}
    nN\ge m[1-H_2(p)].
    \label{eq:s-width-copies}
\end{equation}
\end{proposition}

\noindent\emph{Proof.}
The states $\rho_x^{\otimes N}$ live on $nN$ qubits and, by assumption, form an $m$-bit QRAC on $nN$ qubits with recovery probability at least $p$.  Nayak's lower bound applies directly~\cite{Nayak1999QRAC}.  \hfill$\square$

This argument permits collective measurements across all $N$ copies and does not require the decoder to measure the original designated observables $A_i$.

For comparison, consider the single-coordinate promise problem for a traceless binary observable $A$, $A^2=I$, whose expectation is known only to lie in either $[\eta,1]$ or $[-1,-\eta]$, with $0<\eta\le1$.  Measuring $A$ directly and applying Hoeffding's inequality gives an $O(\eta^{-2})$ upper bound for fixed error probability.  For this expectation-only promise, the same scaling is unavoidable even for arbitrary collective measurements across copies.  Writing $F(\rho,\sigma)=\|\sqrt{\rho}\sqrt{\sigma}\|_1$ for root fidelity, in Hilbert-space dimension $D$ the commuting states
\begin{equation}
    \sigma_\pm=\frac{I\pm\eta A}{D},
    \qquad
    F(\sigma_+,\sigma_-)=\sqrt{1-\eta^2},
\end{equation}
have expectations $\Tr(\sigma_\pm A)=\pm\eta$.  For $N$ copies the fidelity is $(1-\eta^2)^{N/2}$.  By the Fuchs-van de Graaf inequality, the trace distance $\tfrac12\|\sigma_+^{\otimes N}-\sigma_-^{\otimes N}\|_1$ is at most $\sqrt{1-(1-\eta^2)^N}$, so achieving any fixed success probability strictly above $1/2$ in discriminating the two signs requires $N=\Omega(\eta^{-2})$.  Thus the optimal copy complexity of this single-coordinate promise problem is $\Theta(\eta^{-2})$.

For a worst-case target of the complete Majorana encoding, Theorem~1 of the Letter implies that the optimal minimum aligned expectation value is $\tan(\pi/4n)$.  Resolving the sign of the weakest correlator from its expectation value promise therefore requires
\begin{equation}
    N_{\rm coord}=\Omega\!\left(\cot^2\!\left(\frac{\pi}{4n}\right)\right)
    =\Omega(n^2)=\Omega(m),
    \label{eq:s-direct}
\end{equation}
and direct measurement achieves the same scaling.  Estimating all $m=\Theta(n^2)$ correlators separately at this precision and fixed error per coordinate would use $\Theta(mn^2)=\Theta(m^2)$ direct-measurement shots.  A confidence guarantee for the entire recovered string requires separate accounting.  This is deliberately a baseline rather than a global lower bound and is narrower than Eq.~\eqref{eq:s-width-copies}.  Joint measurements, classical shadows, or structured postprocessing may amortize copies across many coordinates or infer bits from additional structure, even though collective measurements do not improve the $\eta^{-2}$ scaling for the isolated single-coordinate promise problem.

\section{Minimax characterization of robust margin}
\label{supp:minimax}
For a fixed target $x$ and observable family $\mathbf A=(A_1,\ldots,A_m)$, define $\delta_x(\mathbf A)=\max_\rho\min_i x_i\Tr(\rho A_i)$.  Let $\mathcal P_m=\{w\in\R_{\ge0}^m:\sum_i w_i=1\}$ be the probability simplex.  Since $\min_i y_i=\min_{w\in\mathcal P_m}\sum_i w_i y_i$, Sion's minimax theorem gives~\cite{Sion1958Minimax}
\begin{equation}
    \delta_x(\mathbf A)
    =\min_{w\in\mathcal P_m}\max_\rho\Tr\!\left[\rho\sum_i w_i x_iA_i\right] =\min_{w\in\mathcal P_m}\lambda_{\max}\!\left(\sum_i w_i x_iA_i\right).
    \label{eq:s-minimax}
\end{equation}
The universal margin of the observable family is $\Delta(\mathbf A)=\min_x\delta_x(\mathbf A)$.  Thus
\begin{equation}
    \Delta(\mathbf A)=\min_{x\in\{\pm1\}^m}\min_{w\in\mathcal P_m}
    \lambda_{\max}\!\left(\sum_i w_i x_iA_i\right).
\end{equation}
The upper-bound half of Theorem~1 can be viewed as an explicit optimal dual certificate $w$ for the transitive tournament.  More generally, Eq.~\eqref{eq:s-minimax} gives an equivalent spectral optimization problem and may permit operator-family-specific bounds sharper than the general QRAC bound.  As a numerical spot check, direct minimization of Eq.~\eqref{eq:s-minimax} for the transitive targets gives $0.41421357$ for $n=2$ and $0.26794964$ for $n=3$.  These agree with $\tan\!\left(\frac{\pi}{8}\right)=0.41421356$ and $\tan\!\left(\frac{\pi}{12}\right)=0.26794919$, respectively, to absolute error below $5\times10^{-7}$.

\section{State-space margin versus simultaneous measurement}
\label{supp:joint}
Tournament matrices also enter the incompatibility analysis of the same complete quadratic-Majorana observable family~\cite{McNulty2025FermionicJoint}.  Despite the common tournament matrix, the two problems depend on different spectral functionals.

For a tournament matrix $T$, let $\nu_1(T),\ldots,\nu_n(T)\ge0$ denote the $n$ nonnegative block magnitudes in its standard skew-normal form.  The $2n$ singular values occur in repeated pairs.  Feasibility of the covariance ray $\Gamma=\tau T$, with $\tau\ge0$, is equivalent to
\begin{equation}
    \tau\,\nu_{\max}(T)\le1,
    \qquad \nu_{\max}(T)=\max_j \nu_j(T).
\end{equation}
Thus the ray guarantees common margin $1/\nu_{\max}(T)$, and the universal worst target is selected by maximizing the \emph{largest} canonical singular value.  The transitive tournament uniquely maximizes this quantity and gives $\nu_{\max}=\cot\!\left(\frac{\pi}{4n}\right)$~\cite{Deng2018Tournament}.

By contrast, the quadratic Hamiltonian $H_T=i\sum_{a<b}T_{ab}\gamma_a\gamma_b$ has operator norm, by Eq.~\eqref{eq:s-quadratic-spectrum},
\begin{equation}
    \|H_T\|_{\op}=\sum_{j=1}^{n}\nu_j(T).
\end{equation}
The incompatibility robustness in Ref.~\cite{McNulty2025FermionicJoint} is controlled by the \emph{sum} of the canonical singular values, equivalently half the skew energy, rather than by $\nu_{\max}$.  Deng \emph{et al.} show that the transitive tournament minimizes this skew energy while maximizing the spectral radius~\cite{Deng2018Tournament}.  Thus the margin-worst class is not the incompatibility-worst class.  The present result limits robust realization of expectation value signs, whereas the joint-measurement result constrains simultaneous measurability.

\end{document}